\documentclass[12pt,psfig]{article}
\usepackage[dvipsnames]{xcolor}
\usepackage[toc,page]{appendix}
\usepackage[toc,page]{appendix}
\usepackage{amsmath, amsthm, amssymb,amsfonts,amscd}
\usepackage{graphicx,epsfig,latexsym}
\usepackage{cite,calc,xcolor,subfigmat,mathrsfs,dsfont}
\usepackage[pagebackref=false]{hyperref}
\usepackage{graphics}
\usepackage{subfigure}
\usepackage{graphicx,epsfig}
\usepackage{amsfonts}
\usepackage{amscd}
\usepackage{mathrsfs}
\usepackage{enumerate}
\usepackage{latexsym}
\usepackage{tikz}
\usepackage[all]{xy}
\newcommand{\Sl}{\mathfrak{sl}_2}
\newcommand{\z}{{{ \mathbb{Z}_2}}}

\newcommand{\Span}{{\rm span}}

\newcommand{\ad}{{\rm ad}}

\newcommand{\ddx}{\frac{\partial}{\partial x}}
\newcommand{\ddy}{\frac{\partial}{\partial y}}
\newcommand{\ddz}{\frac{\partial}{\partial z}}
\newcommand{\ddr}{\frac{\partial}{\partial \rho}}
\newcommand{\ddt}{\frac{\partial}{\partial \theta}}

\newcommand{\IM}{{\rm Im}\,}

\newtheorem{thm}{Theorem}[section]
\newtheorem{cor}[thm]{Corollary}
\newtheorem{lem}[thm]{Lemma}

\theoremstyle{definition}
\newtheorem{defn}[thm]{Definition}
\newtheorem{exm}[thm]{Example}

\newtheorem{notation}[thm]{Notation}
\numberwithin{equation}{section}
\usepackage{color}
\usepackage[top=3cm, bottom=2.5cm,  right=3cm,  left=1.5cm]{geometry}
\newcounter{IssueCounter}
\newtheorem{Issue}[IssueCounter]{Issue}

\def\be {\begin{equation}}
\def\ee {\end{equation}}
\def\ba {\begin{eqnarray}}
\def\ea {\end{eqnarray}}
\def\bpr {\begin{proof}}
\def\epr {\end{proof}}
\def\bes {\begin{equation*}}
\def\ees {\end{equation*}}
\def\bas {\begin{eqnarray*}}
\def\eas {\end{eqnarray*}}

\begin{document}
\baselineskip=18pt
\renewcommand {\thefootnote}{\dag}
\renewcommand {\thefootnote}{\ddag}
\renewcommand {\thefootnote}{ }

\title{
		On the representations and 
  \(\z\)-equivariant normal form for  solenoidal Hopf-zero
  singularities
}
\author{
	Fahimeh Mokhtari\footnote{Corresponding author. Email: fahimeh.mokhtari.fm@gmail.com}
	\\
	Department of Mathematics, Faculty of Sciences\\
	Vrije Universiteit, De Boelelaan 1081a,\\ 1081 HV Amsterdam, The Netherlands
}
\date{}
\maketitle
%
%
%
%

{\bf Dedicated to Professor Jan. A Sanders on the occasion of his 70th birthday }
\begin{abstract}
		In this paper, we deal  with the solenoidal conservative  Lie algebra 
	associated to the classical normal form of Hopf-zero singular system.  We concentrate on the study of some
	representations and \(\z\)-equivariant normal form  for such singular differential equations.
	First, 	we   list some of  the representations that this Lie algebra admits.  
 The vector fields from this Lie algebra could be 
expressed  by the set of ordinary differential equations where the first  two
	 of them  are in the canonical form of a one-degree of freedom Hamiltonian system and the third one
	depends upon the  first two variables. This representation is governed  by   the  associated Poisson algebra   to one sub-family of this Lie algebra. 
	Euler's form, vector potential, and Clebsch representation  are other representations of this Lie algebra that we list here.
	We also study  the  non-potential property  of  vector fields with  Hopf-zero singularity from  this Lie algebra.
	Finally, we examine  the unique  normal form with non-zero cubic terms   of this family 
	 in the  presence  of the  symmetry group \(\z.\)
	The theoretical results of normal form theory   are
	illustrated with the  modified Chua's oscillator.
\end{abstract}
\:\:\:\:\ \ \rule{7.0in}{0.012in}
\section{Introduction}
This investigation is a continuation of \cite{GazorMokhtariInt} in which   
the  maximal solenoidal conservative Lie algebra  of  classical normal form  of Hopf-zero singularities (in the sense of \cite{arnold1971matrices}) was introduced. This Lie algebra was denoted  by \(\mathscr{L}.\) 
The  simplest normal form, simplest parametric normal form, and radius of convergence corresponding to the second level
normal form  of  this type of singularities were explored  there. 
This present paper has two purposes: first,  we intend to list some of   the representations  that 
the vector fields from \(\mathscr{L}\)  admit.
Solenoidal property  of vector fields  from \(\mathscr{L}\)  allows  us  to  express  these  vector fields by
local Euler's potentials,
  vector potential, and local Clebsch potentials, see \cite{mezic1994integrability,barbarosie2011representation}.
 The second purpose is   to  explore the  \(\z\)-equivariant unique  normal form 
 of the  solenoidal conservative family associated to Hopf-pitchfork
 singularities. 
 We denote the Lie algebra  of such a family  
by \(	{\mathscr{L}}^{\z}.\)
Now, 
we 
elaborate  in general terms  our results in the following.

In the framework of classical mechanics,  the Poisson structure  has a key role in the description  of Hamiltonian dynamics
(see preface of \cite{fernandes2014lectures}). 
In  the literature of normal form  theory, this structure is  employed to   facilitate  the normal form study of singularities,
an exhaustive
treatment of this subject can be found in  \cite{BaidSand91,baidersanders}.
The first  issue we address  here   involves constructing the Poisson algebra \(\mathcal{P}\) for
a  Lie subalgebra of  \(\mathscr{L}\), in order to   give a representation for \(\mathscr{L}\).
Having established the Poisson structure we will be able to present  the vector fields from  this Lie algebra
by  a set of ordinary differential equations, the   two first of
them are in the canonical form of a one-degree of freedom Hamiltonian system and the third one
depends upon the two first variables. Noting that,
for a given   three-dimensional solenoidal vector field which  possesses one-parameter symmetry group with solenoidal and infinitesimal generators, one would be able to  present the system  in the normal form as given in \cite{mezic1994integrability}.  
Hence,  beside the Poisson structure 
 following  \cite[Remark 2.4]{GazorMokhtariInt}  and \cite{mezic1994integrability} that form   could be derived.

Solenoidal  dynamical systems  are important in a huge variety of applications.
These vector fields, also known as incompressible vector fields, occur in many settings: velocity field of an incompressible fluid, magnetic field, 
{\em etc.},
see  for instance \cite{priest1996bifurcations,brown1999topological}.
Exploiting  tools from geometric dynamics,
solenoidal vector fields   may be represented in 
some interesting  and physically  significant  representations such as 
Euler's form and  vector potential (see \cite{barbarosie2011representation} and references therein).
In the following, we 
briefly  review these representations.  
If \(v\) be a solenoidal vector field then there exists 
a  vector potential \(\bf A\)  such that \(v=\nabla\times \bf A.\) 
From the point of view of electromagnetism,  \(\bf A\)  is called  magnetic vector potential \cite{semon1996thoughts}.
Euler's form  is another representation that  solenoidal vector fields admit (see \cite[page 22]{truesdell1954kinematics} and 
\cite[page 48]{udriste2012geometric}).
Based on this form,
one can find  two  independent  invariant functions
 \(\alpha\) and \(\beta\) from which the vector field \(v\) is derivable via
\(v=\nabla \alpha \times \nabla \beta.\) 
Applications of these representations can be found in various   
books and papers, for instance, see \cite{aharonov1959significance,barbarosie2011representation,ray1967topological,semon1996thoughts}. In  \cite{barbarosie2011representation} the application of Euler's form in  the partial differential equation 
has been discussed. For the significance of vector potential  in quantum theory,
we refer the reader to  \cite{aharonov1959significance}.

In \cite{GazorMokhtariSandersBtriple}  the authors
express  
a  Lie algebra of completely  integrable solenoidal  triple-zero singularities via  Euler's form and vector potential.  As in their studies,
due  to  the  solenoidal property  of \(\mathscr{L},\)
here we shall present  any 
vector field  in this Lie algebra 
using   vector potentials and Euler's form.
Further, the 
non-potential property of this family with Hopf-zero singularity is examined. Thus, these vector fields cannot be derived by the  gradient of a scalar valued  function, instead, 
these vector fields are  expressible  in terms of the  lamellar and  complex lamellar vector fields.
More precisely, any vector field \(v\in \mathscr{L}\)  could be expressed  by  \(v=f_1\nabla f_2+ \nabla f_3\)
where \(f_1,f_2,\) and \(f_3\) are scalar valued functions.
This representation is named Monge representation or  Clebsch representation (see \cite[page 27]{truesdell1954kinematics} and \cite[Section 2.4]{udriste2012geometric}).
For the  application of this  representation in thermodynamics, we refer to \cite[Subsection 9.11]{udriste2012geometric}.

Progress towards deriving  these representations 
has significant practical implications for magnetic fields  since  these vector fields are solenoidal. 
For instance, in \cite{priest1996bifurcations,brown1999topological}  these representations have been employed for studying magnetic reconnection  at three--dimensional null points.
In  \cite{longcope2005topological}  the authors  used the magnetic fields  to interpret the solar flares.
 Hence, the results of the present paper  should provide  the researchers in this area 
with sufficiently powerful tools 
to analyze  the motion of magnetic fields.
As a matter of fact, each of these representation gives an interpretation about the structure of
magnetic fields.
 Thereby,
from  these researches,
we believe that these studies  would be  remarkably useful in application. 
 In the remainder  of introduction,  we focus  on the second  aim of the paper.

In order to apply the
general methods of bifurcation theory to singularities, it is necessary to  apply  normal form theory.
Roughly speaking,
normal form  theory is to  simplify the nonlinear part of vector fields with permissible transformations, see \cite{MurdBook,SandersSpect}
and
\cite[chapters 9-13]{sanders2007averaging}.
In this paper, we are  also interested in treating the simples normal form classification of  solenoidal family   associated to  the  classical normal form of 
Hopf-pitchfork singularities.
The problem of the  normal form  of singular dynamical systems has been studied by
many authors.
Before explaining  our results in detail, we explain
some of this previous work in the following. 

Baider and Sanders \cite{BaidSand91,baidersanders} studied  the unique  normal form of   Bogdanov-Takens  and  Hamiltonian    Bogdanov-Takens  singularities.
The paper \cite{GazorMokhtariSandersEulHZ}  studied the infinite level  normal form of the
Lie algebra of  quasi-Eulerian 
Hopf-zero vector fields.
For    the  conservative-nonconservative
decomposition of the classical normal form of Hopf-zero dynamical systems and complete classifications of the  simplest normal form  of this  singularity, the reader is referred to \cite{GazorMokhtari}. 
In \cite{AlgabaHopfZ} the \(\z\)-equivariant  normal form for Hopf-zero vector fields are computed 
 under the assumption that 
the  cubic terms be non-zero, for results on the bifurcation of  Hopf-pitchfork singularities, see \cite{algaba2000tame,algaba1999three,algaba1999codimension}. More  studies regarding the normal form of dynamical system could be found in  \cite{GazorMoazeni,Algaba30,GazorMokhtariSandersBtriple,Murdn}.

In the studies mentioned above and in the literature, the researchers did not investigate the problem of  classifying the unique  normal form for
\(\z\)-equivariant
solenoidal conservative Hopf-zero vector fields. 
As  announced at the beginning of  this section,
our   aim  is     to treat the unique  normal form  of the following system
\begin{equation}
\label{CNF}
\left\{
\begin{aligned}
\frac{{\rm d}{x}}{{\rm d }t}&=2 a^0_1 x\rho^2+\sum a^l_k (k-{2l}+1){x}^{{2l}+1}{\rho}^{2(k-{2l})},
\\
\frac{{\rm d}{\rho}}{{\rm d }t}&=-\frac{a^0_1\rho^3}{2}-\sum a^l_k \frac{({2l}+1)}{2} {x}^{2l}{\rho}^{2(k-{2l})+1},
\\
\frac{{\rm d}{\theta}}{{\rm d }t}&=1+\sum  b^l_k{x}^{{2l}}\rho^{2(k-{2l})},
\end{aligned}
\right.
\end{equation}
where  constants  \(a^l_k, b^l_k\) are real numbers,  \(0\leqslant2l\leqslant k,  0\leqslant k\) \(b^0_0=a^0_0=0,\) and \( a^0_1\neq 0.\)
The associated  first integral of this system  is 
\bas
s(x,\rho)&:=& a^0_1x\rho^4 +\sum a^l_k x^{2l+1} {\rho}^{2({k-2l+1})}, 
\eas 
with \(0\leqslant2l\leqslant k.\) 
In this work, the   unique normal form  study of the above system proceeds in a manner parallel to the study of the  solenoidal 
Hopf-zero
vector fields 
without symmetry \cite{GazorMokhtariInt}. 
 We  use  naturally  their  algebraic structures  to construct the Lie algebra \(	{\mathscr{L}}^{\z}.\)
We wish to stress that, despite the fact that  \(	{\mathscr{L}}^{\z}\) is  the Lie subalgebra of \(\mathscr{L}\), 
the unique  normal form of \eqref{CNF}
cannot be obtained
from  the unique normal form of  volume preserving Hopf-zero  vector fields  given there.
The leading term that plays a dominant role in our  normal form study  is  cubic, whereas there 
the quadratic term was the leading term. 
Hence,  the normalization  problem that is performed here differs from  that was  studied  in \cite{GazorMokhtariInt}. 
Finally, we shall  present  the  unique normal form  of \eqref{CNF} in  four  different  representations, see Theorem \eqref{4rep}.
\subsection{Outline of the paper}
This paper has the following organization. In Section \ref{repres}, first,  we  recall  some  notations and definitions
that   needed throughout the paper. Then, 
we provide
the associated Poisson algebra \(\mathcal{P}\)  for  one Lie subalgebra  of  \(\mathscr{L}.\) Following the  Poisson algebra,  a form for solenoidal Hopf-zero vector fields
using Hamilton's equations  as introduced in \cite{mezic1994integrability} is presented.
It is also
in this section that 
the
representations such as 
Euler's form,  vector potential,  and Clebsch representation for Lie algebra \(\mathscr{L}\) are given.

In Section  \ref{HZFO},  we introduce and formulate the  
Lie algebra 	\(	{\mathscr{L}}^{\z}.\)
We also recall the general framework for computing  the normal form for
\({	\Gamma}\)-equivariant
singularities required to  study  the unique  normal form of  the class of singularities  under consideration.

In Section \ref{NFC}, we study   the  unique  normal form   of 
\eqref{CNF}. One symmetry of unique normal form is detected. We  also  present  the four alternative representations of  the unique normal from rely on the results on  Section \ref{repres}.

The final section
is  dedicated to make  symbolic computations  of our normal form study. Some sufficient  conditions on the coefficients of any Hopf-pitchfork system,  under which the lower order truncation of  the    classical normal form  of the  original system 
takes the form \eqref{CNF} are given.
Moreover, the modified Chua's oscillator serves to demonstrate  our  main results in the normal form. All  of the  computations are  performed using Maple \cite{char2013maple}.
\section{ Representations  of Lie algebra \(\mathscr{L}\)}\label{repres}
\subsection{Preliminaries}
In this section, we   present the  main results regarding the ways that the vector fields from  \(\mathscr{L}\)
would be expressible. The  Poisson algebra associated to the sub-family of \(\mathscr{L}\) and tools from geometric dynamics are adapted to  
obtain  these forms.
Before going to the  main results, we  give a review of some definitions and facts from  the Lie algebra \(\mathscr{L}\) which are fundamental to what follows. Furthermore, after   we establish some notation,  we recall some  general information   required 
to  study the  representations of vector fields  in \(\mathscr{L}.\)

 The  classical normal form of solenoidal conservative Hopf-zero vector field in cylindrical coordinates   is  given by 
\begin{equation}
\label{Eq1}
\left\{
\begin{aligned}
\frac{{\rm d}{x}}{{\rm d }t}&= \sum_{} (k-l+1)a^l_k x^{l+1} {\rho^2}^{(k-l)}, 
\\
\frac{{\rm d}{\rho}}{{\rm d }t}&= -\sum_{} \frac{(l+1)}{2}a^l_kx^{l} {\rho}^{2(k-l)+1},
\\
\frac{{\rm d}{\theta}}{{\rm d }t}&= \sum_{} b^m_nx^{m} {\rho}^{2(n-m)}, 
\end{aligned}
\right.
\end{equation}
where \(-1\leqslant l\leqslant k, 0\leqslant k,\)  \(0\leqslant m\leqslant n,\) \(b^{-1}_0=a^0_0=0,\) and \(a^l_k, b^m_n\in \mathbb{R},\)
see \cite[Equation 1.1]{GazorMokhtariInt}. 
This class of vector fields was derived  by 
\(\Sl\)-decomposition of the  classical normal form of Hopf-zero  bifurcation 
\cite{GazorMokhtari}.
َAnd indeed,  the above system could be formulated using the Lie algebraic structure as follows.

  We recall\cite{GazorMokhtariInt,GazorMokhtari}  that   the   maximal Lie algebra  of solenoidal Hopf-zero  classical normal form
is given by 
\(
{\mathscr{L}}=\mathscr{F}\oplus\mathscr{T},
 \)
where
\bas
\mathscr{F}=\Span\left\{ \sum a^l_k F^l_k\,|\,0\leqslant k,-1\leqslant{l}\leqslant{k}, a^l_k\in \mathbb{R}\right\},
\quad
\mathscr{T}=\Span\left\{ \sum b^l_k \Theta^l_k\,|\,0\leqslant{l}\leqslant{k}, b^l_k\in \mathbb{R}\right\},
 \eas
and
\ba
\label{FT}
F^l_k&= &x^{l} {(y^2+z^2)}^{k-l}\left((k-l+1) x\ddx-
\frac{(l+1)}{2} y\ddy-\frac{(l+1)}{2} z\ddz\right),\quad   -1\leqslant{l}\leqslant{k},\;\;\;
\\\label{TTerm}
\Theta^l_k&=& x^{l}(y^2+z^2)^{k-l}\left(z\ddy-  y\ddz\right),\qquad
\qquad\qquad\qquad\qquad\qquad\qquad\,\,\,\,\,\,\;\;
0\leqslant{l}\leqslant{k}.\quad\;\;\;
\ea
Furthermore,  the algebra of the first integral for \(\mathscr{F}\) is as
\(
\left\langle x^{l+1}(y^2+z^2)^{k-l+1}\right\rangle_{-1\leqslant{l}\leqslant{k}}^{ 1\leqslant l+k}\)
 and for  \(\mathscr{T}\)
is as  \(
\left\langle x,(y^2+z^2)\right\rangle
\).
 Expressed in terms of the  cylindrical polar   coordinates  \((x, \rho, \theta),\) the  preceding vector fields  
will be given by the expressions
\ba\label{cylindFT}
F^{l}_k&=&  x^{l}\rho^{2(k-l)}\left( (k-l+1)x\ddx- \frac{(l+1)}{2} \rho\ddr\right),
\\\label{cylindTT}
\Theta^l_k&=&{x}^{{l}}\rho^{2(k-{l})}\ddt.
\ea

Using  the structures  given above,  
system \eqref{Eq1}  can be  recast  to   
\bas
v=\sum_{k=0}^{\infty}\sum_{l=-1}^k a^l_k F^l_k+\sum_{n=0}^{\infty}\sum_{m=0}^n b^m_n \Theta^m_n.
\eas
As mentioned before, the simplest normal form and simplest parametric normal form of foregoing system  with the assumption  \(a^{-1}_0\neq 0\) were explored
using \(\Sl\)-style in
\cite{GazorMokhtariInt}.  Now, we fix some notation.
\begin{notation}
	The following notation  is used throughout the paper.
	\begin{itemize}
	\item   Define \(v=(v_1,v_2,v_3)\in \mathbb{R}^3\) by \(v=v_1\cdot\mathbf{ {e}_{x}}+v_2\cdot\mathbf{ {e}_{y}}+v_3\cdot\mathbf{ {e}_{z}}.\)
	We denote
 \(F^l_k=dx(F^l_k)\ddx+dy(F^l_k)\ddy+dz(F^l_k)\ddz.\)
	\item The symbol  \(\nabla\)
	indicates the  gradient operator of  the vector field.
	\item  The Pochhammer \(k\)-symbol notation for any \(a,b\in\mathbb{R}\) and \(k\in \mathbb{N},\) is given by
	\bas
	\left(a\right)^k_b:=\prod_{j=0}^{k-1}(a+jb).
	\eas
\end{itemize}
\end{notation}
\begin{defn}\label{Geom}
	Consider  the dynamical system \({\dot x}=v,\) with \(x\in \mathbb{R}^3.\)  	
	\begin{itemize}
		\item 
			The vector field \(v\) is said  to be  potential, or locally potential
		if \(\nabla \times v=0\) for all \(x\in \mathbb{R}^3.\) 
		The system
		is named  non-potential,
		otherwise,
		 see \cite[page 1]{tarasov2008quantum}.
		\item
			A vector field \(v\)  which derives from  the  gradient  of a   function is called 
lamellar (also known as gradient, or globally potential)  vector field. The  function is called the  potential function.
It follows  that a vector field is  lamellar if and only if
the system  is   potential,
see \cite[page 23]{truesdell1954kinematics} and  \cite[page 1]{tarasov2008quantum}.
\item
The  three-dimensional vector field \(v\) which has  the  representation of the  form  \(v=f \nabla g,\)
in  which  \(f\) and \(g\) are   functions,
 is named complex lamellar.
The vector field \(v\) is complex lamellar  if and only if  \(v\cdot(\nabla \times v)=0,\) that is, this 
type of vector field  is orthogonal to its curl, see \cite[page 23]{truesdell1954kinematics}.
	\end{itemize}
\end{defn}
\subsection{Poisson structure}
As  mentioned, the  algebra of the first integral for \(\mathscr{F}\) is spanned by 
  \(
  \left\langle x^{l+1}(y^2+z^2)^{k-l+1}\right\rangle_{-1\leqslant{l}\leqslant{k}}^{ 1\leqslant l+k}
 . \) 
 In what follows, we 
extend this  algebra to  the  Poisson structure  for \(\mathscr{F}\) by equipping it with an  appropriate Poisson bracket.
In accordance with  this Poisson  algebra, one representation of any  vector fields in  \(\mathscr{F} \) by the  Hamiltonian equation is given.

Performing the change of variable \(r=\rho^2\) into   the vector field   given by \eqref{cylindFT}  we obtain 
\ba\label{fcylr}
F^{l}_k&=& x^{l} r^{k-l} \left((k-l+1) x\ddx-(l+1) r\frac{\partial}{\partial r}\right).
\ea
 Now, define
 \ba\label{pl}
 \mathcal{P}:=\left\{\sum c^l_k  f^l_k\mid -1\leqslant l\leqslant k \right\},
 \ea
 where
 \(
 f^l_k:= x^{l+1}r^{k-l+1} 
 \) is the  first integral of \(F^l_k.\)   Now, we have  the following result.
 \begin{thm}
	\begin{itemize}
		Consider the vector field \(F^l_k\)  given by \eqref{fcylr} and the algebra \(\mathcal{P}\)   given by \eqref{pl}.
		The following statements  hold.
		\item [1.]  \(({\mathcal{P}}, \{\cdot, \cdot\})\) is Poisson algebra where  the Poisson bracket is given by 
\be\label{poissonb}
\{f,g\}:= {\frac{\partial}{\partial r}}  f \ddx g-  {\frac{\partial}{\partial r}}  g \ddx f, \qquad {\hbox{for all}}\,\, f,g\in \mathcal{P}.
\ee
		\item [2.] 	Hamilton's equations of \(F^l_k\)  are 
		\ba\label{HamF}
		F^l_k=\{f^l_k,x\}\ddx+\{f^l_k,r\}\frac{\partial}{\partial r}.
		\ea
				\item [3.]  \(({\mathcal{P}}, \{\cdot, \cdot\})\) and \((\mathscr{F}, [\cdot,\cdot])\) are  isomorphic Lie algebras
where  Lie isomorphism is defined  by 
	\(
\varphi: ({\mathcal{P}}, \{\cdot, \cdot\})\rightarrow (\mathscr{F}, [\cdot,\cdot]),\)  with
\(\varphi({f}^l_k)=F^l_k.
\)
\end{itemize}
	\end{thm}
\bpr
One can  readily check that  \(({\mathcal{P}}, \{\cdot, \cdot\})\) is a Poisson  algebra 
(see  \cite[Section 5]{FMTHESIS} and 
\cite[Chapter one ]{fernandes2014lectures}). 
A direct calculation using the Hamiltonian function given by the  first item and  Poisson bracket given by  \eqref{poissonb}
establish \eqref{HamF} (see  also preface of  \cite{fernandes2014lectures}).
The last  claim follows 
from the  straightforward  calculation  of structure constants given by 
\eqref{poissonb} and \cite[Lemma 2.5]{GazorMokhtariInt}.
\epr   
\begin{cor}\label{Hamiltonian}
The class  of Hopf-zero singular system 
given by \eqref{Eq1} may be represented  through appropriate  coordinates transformations 
to the form
\begin{equation}
\label{Hamiltoniane}
\left\{
\begin{aligned}
\frac{{\rm d}{x}}{{\rm d }t}&
= \frac{\partial H(x,r)}{\partial r}=\{H(x,r),x\},\\
\frac{{\rm d}{r}}{{\rm d }t}&=-\frac{\partial H(x,r)}{\partial x}=\{H(x,r),r\},
\\
\frac{{\rm d}{\theta}}{{\rm d }t}&=G(x,r),
\end{aligned}
\right.
\end{equation}
	where \(
	H(x,r):= \sum a^l_k x^{l+1} {r}^{{k-l+1}}
	\)    is a constant
	of the motion  and  \(G(x,r):=1+\sum  b^m_n{x}^{{m}} r^{n-m}.\)
		\end{cor}
	\bpr
	The corollary can be verified in two ways:
	\( (1)\)
	recast the dynamical system \eqref{CNF} by  employing   the coordinate
	change \(r=\rho^2.\) 
		Then,
 taking  Hamilton's equations corresponding to \(F^l_k\)  given by Equation \eqref{HamF}
into account,
 immediately verify the statement of the corollary. Then 
\((2)\)
 follows from  \cite[Remark 2.4]{GazorMokhtariInt} and 
  the 
  procedure   given in the proof of \cite[Theorem 2.2]{mezic1994integrability}
  with the slight modification. In fact,   replace \(J\) by \( 2J,\) into that proof.
   First by employing  the  change of variables  \(y= \rho\cos(\theta), z=\rho \sin(\theta),\)  and \(x=x,\)
		the   differential equation \eqref{Eq1}  goes over into
	 the system
	 \begin{equation*}
	 \left\{
	 \begin{aligned}
	 \frac{{\rm d}{x}}{{\rm d }t}&= \frac{\partial K(x,\rho)}{2J\partial \rho},
	 \\
	 	\frac{{\rm d}{\rho}}{{\rm d }t}&=-\frac{\partial K(x,\rho)}{2J\partial x},
	 	\\
	 		\frac{{\rm d}{\theta}}{{\rm d }t}&=G(x,\rho),
	 	 \end{aligned}
	 \right.
		 \end{equation*}
	 where \(J\!=\rho\)   is the Jacobian of the  cylindrical transformation,
	 \(K(x,\rho):=\sum a^l_k x^{l+1} {\rho}^{2({k-l+1})},\) and  \(G(x,\rho):=1+\sum  b^m_n{x}^{{m}}\rho^{2(n-m)}.\)
	 To proceed further,   carrying  out the   transformation \(r=\int 2J=\rho^2\) into the  former system.
	  This 
	 turns  the above system into
	   the form   given  by \eqref{Hamiltoniane}, see also \cite[Example 3]{mezic1994integrability}.
 The corollary follows.
 		\epr
\subsection{Euler's form for \(\mathscr{L}\)}\label{EulerForm}
Any solenoidal vector field may be represented   by   Euler's form  as
\ba\label{euler}
v=h(g_1,g_2)\nabla g_1\times \nabla g_2,
\ea
where  \(h,g_1,\) and \(g_2\) are scalar valued functions. We remark that \(g_1\) and \(g_2\) may not be defined everywhere on the domain of \(v\). The functions 
\(g_1\) and \(g_2\) are 
 first integrals or vector sheets  of vector field \(v\)
and  named  Euler's potentials.
The geometrical meaning of Euler's form  is that 
the vector lines  of solenoidal vector fields are the intersection of level surfaces of Euler's potentials, see \cite[page 22]{truesdell1954kinematics},  
Euler's Theorem
\cite[page 48]{udriste2012geometric}, and \cite{longcope2005topological}.

In the following result, we  express the vector fields in \(\mathscr{L}\) by local Euler's potentials.
\begin{thm}\label{Eulerform}
	The Euler's form of \({ F}^l_k\) and \(\Theta^l_k\) are given as follows
	\ba\label{ClH}
	{F}^l_k&=& \nabla
	\left( \frac{1}{2}\arctan ( {\frac {z}{y}} )\right)\times \nabla \left(x^{l+1} (y^2+z^2)^{k-l+1}\right),\qquad \hbox{for}\quad  y\neq 0,
	\\\label{ClT}
	{	{\Theta}^l_k}&=&
\nabla  \left(x^{l+1}\right)\times \nabla \left({\frac { -\left( {y}^{2}+{z}^{2} \right) ^{k-l+1}}{2 \left(l+1
		\right)\left( k-l+1 \right) }}\right) .
	\ea
\end{thm}
\bpr
We prove this theorem  only for \({F}^l_k.\)
The proof 
 for \(\Theta^l_k\)
 can be done analogously.
Due to the solenoidal property  of vector fields  in  \(\mathscr{L}\) and Euler's Theorem
\cite[page 48]{udriste2012geometric},
  \(F^l_k\) may be written  as \eqref{euler}.
Setting the   first integrals of  \(F^l_k\) 
as local Euler's potentials, in fact,  \(g_1:=x^{l+1}(y^2+z^2)^{k-l+1}\) and \(g_2:=\frac{1}{2}\arctan ( {\frac {z}{y}}).\)
Now,  we show  that \(h(f_1,f_2)\)  given  by \eqref{euler} equals one. By straightforward calculation, one has
\bas
\nabla\left(\frac{1}{2}\arctan ( {\frac {z}{y}} )\right)&=&{\frac {-z}{2({y}^{2}+{z}^{2})}}\cdot\mathbf{ {e}_{y}}+{\frac {y}{2({y}^{2}+{z}^{2})}}\cdot\mathbf{ {e}_{z}},
\\
\nabla\left(x^{l+1} (y^2+z^2)^{k-l+1}\right)&=&
x^{l}(y^2+z^2)^{k-l+1}\Big((l+1)(y^2+z^2) \cdot\mathbf{ {e}_{x}}+2(k-l+1)x y\cdot\mathbf{ {e}_{y}}
\\&&
+2(k-l+1)xz \cdot\mathbf{ {e}_{z}}\Big).
\eas
By taking the cross  product of  the 
  above vector fields and  considering \eqref{FT} one finds
\bas
 \nabla
\left(\frac{1}{2}\arctan ( {\frac {z}{y}})\right)\times\nabla \left(x^{l+1} (y^2+z^2)^{k-l+1}\right)= 
F^l_k.
\eas
This  implies that  \(h(f_1,f_2)=1\) and concludes our assertion.
\epr
\subsection{Vector potential for \(\mathscr{L}\)}
Another   way that  solenoidal vector field   could be  constructed,  is by mean  of the  vector potential.
These vector fields are expressible   from another vector field by taking its curl.
To find  more  information about  the  physical significant of  vector potential, see \cite{konopinski1978electromagnetic,phatak2010three,tonomura1987applications}.
\begin{thm}\label{vectpot}
	The vector  fields  \({F}^l_k\) and \(\Theta^l_k\) could be  expressed in the following form.
	\bas
	{F}^l_k&=&\nabla\times \left( \frac{1}{2}x^{l+1} (y^2+z^2)^{k-l}({ {-z}}\cdot\mathbf{ {e}_{y}}+{ {y}}\cdot\mathbf{ {e}_{z}})\right),
	\\
	\Theta^l_k&=&\nabla\times \left({\frac {-{x}^{l} \left( {y}^{2}+{z}^{2} \right) ^{k-l+1}}{2(k-l+1)}} \cdot\mathbf{ {e}_{x}}\right).
	\eas
\end{thm}
\bpr
The proof follows from
\cite[Equation 5]{longcope2005topological} and Euler's form given by Theorem \ref{Eulerform}.
\epr
The vector potential is not unique since the curl of a gradient is zero.
In the following theorem  we use 
an alternative approach
to construct a  vector potential associated  to the  solenoidal vector fields \(F^l_k\) and \(\Theta^l_k.\)
\begin{thm}
	The vector fields \(F^l_k\) and \(\Theta^l_k\) can be  produced  in  the following form as
	\ba\label{VPF}
	 F^l_k&=&\nabla\times {\bf A}^l_k,
	\\\label{VPT}
	\Theta^l_k&=& \nabla\times {\bf B}^l_k,
	\ea
	where the vector potentials  \({\bf A}^l_k\) and \({\bf B}^l_k\) are defined as follows
		\bas
	{\bf A}^l_k&:=&	 \sum_{j=0}^{k-l} {\frac {{y}^{2j}
			{z}^{2(k-l-j)+1} {x}^{l}{k-l\choose j}}{2(k-l-j)+1}}\left(\frac{l+1}{2}y \cdot\mathbf{ {e}_{x}}+\left( k-l+1 \right) {x} \cdot\mathbf{ {e}_{y}}
	\right),
	\\
	{\bf B}^l_k&:=&\frac{{x}^{l}}{{2(k-l+1)}}\left({ { ( {y}^{2}+{z}^{2} ) ^{k-l+1}}}
	-{{{y}^{2(k-l+1)}}}\right)\cdot\mathbf{ {e}_{x}}-{\frac {{x}^{l+1}{y}^{2k-2l+1}}{l+1}}\cdot\mathbf{ {e}_{y}}.
	\eas
\end{thm}
\bpr
 To prove the theorem,  
 we follow the approach  given in 
 \cite{MichaelSullivan}. 
 By solving   the equality \(F^l_k=\nabla\times {{\bf A}^l_k},\)  we obtain the following equalities
 \ba\label{EEE1}
 \ddz{{\bf A}^l_k}	\cdot\mathbf{ {e}_{x}}= dy(F^l_k),\qquad
-\ddz{{\bf A}^l_k}	\cdot\mathbf{ {e}_{y}}= dx(F^l_k),\qquad
\ddx{{\bf A}^l_k}	\cdot\mathbf{ {e}_{y}}-\ddy{{\bf A}^l_k}	\cdot\mathbf{ {e}_{x}}=
dz(F^l_k).
\ea
The first identity of \eqref{EEE1} implies that 
\ba\label{EE2}
{{\bf A}^l_k}\cdot\mathbf{ {e}_{x}}&=& -\frac{l+1}{2}x^{l} \sum_{j=0}^{k-l} \frac{{k-l\choose j}{y}^{2j+1}
	{z}^{2(k-l-j)+1}}{2(k-l-j)+1}+p_1(x,y),
\ea
analogously  the second identity of \eqref{EEE1}  implies that 
\ba\label{EE1}
{{\bf A}^l_k}\cdot\mathbf{ {e}_{y}}=- (k-l+1)x^{l+1} \sum_{j=0}^{k-l}
\frac{{k-l\choose j}{y}^{2j}
	{z}^{2(k-l-j)+1}}{2(k-l-j)+1}+p_2(x,y),
\ea
where \(p_1(x,y)\) and \(p_2(x,y)\) are   functions.
Now by substituting  \eqref{EE2} and \eqref{EE1} into the last  identity of \eqref{EEE1} one can find, 
\(\ddz p_1(x,y)-\ddx p_2(x,y)=0.\)
Thus, without loss of generality we may assume that  \( p_1(x,y)= p_2(x,y)=0,\) the 
equality \eqref{VPF} concludes.
By repeating this procedure 
for  \(\varTheta^l_k\) we can verify  equality \eqref{VPT}.
\epr
\subsection{Clebsch representation for \(\mathscr{L}\)}
In this part,  we shall study  the following problems: namely, we shall  study 
the  non-potential property of Hopf-zero singularities from \(\mathscr{L},\) and we shall present  an alternative representation that is  the Clebsch representation  of  any vector fields from \(\mathscr{L}.\)
Based  on this   representation
for a given  vector field \(v\in \mathbb{R}^3,\) 
for any  point \(x\in \mathbb{R}^3\)
   in which 
\(\nabla\times v\neq 0,\)
  one can find three 
scalar valued functions \(f_1,f_2,\) and \(f_3\) such that 
\ba\label{monge}
v= f_1\nabla f_2+\nabla f_3.
\ea
These scalar valued  functions are called Clebsch potentials of \(v.\) Noting  that \(f_1,f_2\) and \(f_2\) may not be defined everywhere on the domain of \(v.\)
Moreover,
this representation shows that   any vector fields may be presented  by summation of  the   lamellar field  and  complex lamellar field,
see 
Definition \ref{Geom}.
For those interested in knowing  that  how this representation  is constructed, we include the following discussion from  the proof of  Clebsch's Theorem  of  \cite[Section 2.4]{udriste2012geometric}.

Since  \(\nabla\times v\) is a solenoidal vector field then,  there exist  local Euler's potentials \(f_1\) and \(f_2\)
such that \(\nabla\times v=\nabla f_1\times\nabla f_2.\)
Making use of  the  vector calculus identity  
we have  
 \(\nabla\times (v-f_1\nabla f_2)=0.\) Hence,
 \(v-f_1\nabla f_2\) is 
irrotational vector field,  then there exists potential function \(f_3\) such that 
\(v-f_1\nabla f_2=\nabla f_3,\) see Definition \ref{Geom}. Then,  the relation \eqref{monge} follows.

For  the Euler's  form  (and consequently  the  representation given by vector potential)  to be possible,
the vector field \(v\) has to be solenoidal. 
We remark that
deriving  the vector field
\(v\) from Clebsch  potentials \(f_1,f_2,\) and \(f_3\) through the relation 
\eqref{monge}, is not related to the solenoidal  property  of vector field \(v.\)
In fact,
any vector field may be represented  
in the Clebsch  representation, see  \cite[Section 2.4]{udriste2012geometric}.

\begin{thm}
Any   \(v\in \mathscr{L}\) with  Hopf-zero singularity is the non-potential vector field.
\end{thm}
\bpr
Recall  from \cite{GazorMokhtariInt}, the grading function for generators  of \(\mathscr{L}\)
given  by \(\delta({\rm F}^l_k)=\delta({\Theta}^l_k)=k.\) 
Suppose that   the non-zero Hopf-zero vector field \(v=\sum_{j=0}^{\infty}v_j\in \mathscr{L}\)  where \(\delta(v_j)=j\) is given.
Assume  that  the  claim of the  theorem  does not hold.
Similar to argument given in \cite{GazorMokhtariSandersBtriple},  
  since  the vector fields with    different grades  do not  have any  monomial in common, it implies that  
\(\nabla \times v=0\) if and only if \(\nabla \times v_j=0,\) for all \(j\in \mathbb{N}_0.\)
  Thus  without loss of generality, it suffices to show that \(\nabla\times v_j= 0.\)   
  Taking into account the defined grading function, the vector field \(v_j\) may be represented by 
   \[v_j:=\sum_{i=-1}^{\lfloor\frac{j}{2}\rfloor} a_i {F}^i_{j}+\sum_{i=0}^{\lfloor\frac{j}{2}\rfloor}b_i\Theta^i_j,\]
   where \(a_i\) and \(b_i\) are real constants for all \(i.\)
 Applying  the curl operator to \(v_j\)  yields
\bas
\nabla\times v_j&=&
\sum_{i=-1}^{\lfloor\frac{j}{2}\rfloor}
(0
,0, a_i\frac{i({i}+1)}{2}{x}^{{i}-1}{\rho}^{2(j-{i})+1}+ 2a_i(j-i)(j-i+1) {x}^{{i}+1}{\rho}^{2(j-{i})-1}
)
\\&&
-\sum_{i=0}^{\lfloor\frac{j}{2}\rfloor}\left(
i b_i {x}^{{i}-1}{\rho}^{2(j-{i})}
,0,
2(j-i+1)b_i {x}^{{i}}{\rho}^{2(j-{i})-1}\right).
\eas
 Hence, the only way that the foregoing relation vanishes 
   is either \(v_j=0\)  or \(v_0=  x\ddx- \frac{1}{2} \rho\ddr.\) This implies that either \(v=0\) or 
the vector field is not    Hopf-zero singularity,
which are  in contradiction 
 to   our assumption.
   Hence, the claim is proved.
\epr
In the terminology
of  Definition \ref{Geom}, this result shows  that  the solenoidal family of Hopf-zero singularities are not lamellar  vector fields.
In this  sense, these vector fields can not be expressed by the  gradient of the  scalar valued  function. 
In what follows,  the Clebsch representation of \(\mathscr{L}\)  is given.
\begin{thm}\label{Clebschtheorem}
	\begin{itemize}
		The following hold.
		\item 
			For given \(F^l_k\in \mathscr{F}\) there exist Clebsch potentials  \(f_1,f_2\), and \(f_3\) 
			 such that
			\ba\label{ClebschF}
			{ F}^l_k=f_1\nabla  f_2+\nabla f_3,
			\ea
			where  for \(l\neq 0,\) 
						\ba\label{f1f2f3}
						\nonumber
						 f_1&:=&x^{l},
						 \\
						f_2&:=& -\frac{\left( l+1
							\right)}{4({k-l+1})}{ { \left( {y}^{2}+{z}^{2} \right) ^{k-l+1}  }}-\frac{\left( k-l+1 \right)}{l}{ {  {x}^{2} \left( {y}^{2
								}+{z}^{2} \right) ^{k-l}}},
									\\\nonumber
		f_3&:=&\frac{1}{l} { { \left( {y}^{2}+{z}^{2} \right) ^{k-l} \left( k-l+1 \right) {x
				}^{l+2}}}.
			\ea
Otherwise
\bas
f_1:=-x,\quad f_2:=(k+1)x (y^2+z^2)^{k},\quad f_3:=-(k+1)x^2(y^2+z^2)^{k}+\frac{1}{4(k+1)}(y^2+z^2)^{k+1}.
\eas
						
		\item
		For each vector field \(	{ \Theta}^l_k\in \mathscr{T}\) there exist local  Clebsch potential \(g_1\)  and   global  Clebsch potentials \(g_2,g_3\)  as follows
			\bas
		g_1:={\frac {z}{y}},\,\,\,\qquad g_2:={y}^{2}{x}^{l} \left( {y}^{2}+{z}^{2} \right) ^{k-l},
		\qquad g_3:=-yz {x}^{l}\left( {y}^{2}+{z}^{2} \right)^{k-l},\qquad \hbox{for}\quad y\neq 0,
		\eas
					such that
	\[
		{ \Theta}^l_k= g_1\nabla g_2+\nabla g_3.
	\]
	\end{itemize}
\end{thm}
\bpr
Following  the  discussions at the beginning of this  part, 
first we need  to
 find the  Euler's form  for \(\nabla \times F^l_k.\)
Writing
\(\nabla \times{\rm F}^l_k,\) given by 
 the preceding theorem 
using   Cartesian coordinates  gives
\bas
\nabla \times F^l_k&=&\left( {y}^{2}+{z}^{2} \right) ^{k-l-1}{x}^{l-1} \left( \frac{l(l+1)}{2} 
\left( {y}^{2}+{z}^{2} \right) + 2
\left( k-l \right)  \left( k-l+1 \right){x}^{2}\right)(z\cdot\mathbf{ {e}_{y}}-y \cdot\mathbf{ {e}_{z}}).
\eas
Suppose that \(l\neq0,\) then by straightforward calculation
	 one  can  verify that \(\nabla \times F^l_k=\nabla f_1\times\nabla f_2,\) where \(f_1,f_2\)  are  given by Equation \eqref{f1f2f3}.
In order that \eqref{ClebschF} to be true the following equalities must hold
\ba\label{F1}
&&
dx(F^l_k)-f_1\nabla f_2\cdot\mathbf{ {e}_{x}}=\nabla f_3\cdot\mathbf{ {e}_{x}},
\\&&
\label{F2}
dy(F^l_k)-f_1\nabla f_2\cdot\mathbf{ {e}_{y}}=\nabla f_3\cdot\mathbf{ {e}_{y}}.
\ea
Equation \eqref{F1} follows  after some computations that 
\bas
 f_3={x}^{l}\left({y}^{2}+{z}^{2}\right)^{k-l}\left({\frac{ \left( k
		-l+1\right){x}^{2}}{l+1}}-{\frac{{y}^{2}+{z}^{2}}{4(k-l+1)}}
\right)+h(y,z).
\eas
Substituting 
  \(f_3\) in \eqref{F2}
 results in \(h(y,z)=0.\)  In  this way, we obtain the expression \eqref{ClebschF}. 
 If \(l=0\)  one can check that 
 \bas
 F^{0}_k&=&(y^2+z^2)^{k}\left((k+1)x\cdot\mathbf{{e}_{x}}-\frac{1}{2}y\cdot\mathbf{{e}_{y}}-\frac{1}{2}z\cdot\mathbf{{e}_{z}}\right)
 \\
 &=&
 -x\nabla\left((k+1)x (y^2+z^2)^{k}\right)
 +
 \nabla\left(-(k+1)x^2(y^2+z^2)^{k}+\frac{1}{4(k+1)}(y^2+z^2)^{k+1}\right).
 \eas 
 The proof  for  \(\Theta^l_k\) is analogous to the proof of \({F}^l_k.\)
 \epr

\section{Solenoidal conservative  \(\z\)-equivariant  Lie algebra}\label{HZFO}
In this section, we shall introduce the  solenoidal  conservative Lie algebra associated to the   classical normal form of Hopf-pitchfork singularities \eqref{CNF}.
We also recall  the  theory of unique normal form  at the end  of this part.

To start, we recall the following definition from \cite[Chapter XII ]{golubitsky2012singularities}.
\begin{defn}
	Let
	\(
\dot{x}=v
	\) with \(x\in \mathbb{R}^n,\)
	be an autonomous dynamical system where
  \(v,\) is smooth. Let \({	\Gamma},\) be
	a compact Lie group in \({\rm Gl}(n)\). 
	This system   is called \({	\Gamma}\)-equivariant if
	\(
	v(\gamma x)=\gamma v
	\)
	for all \(\gamma\in{	\Gamma}\) and  \(x\in \mathbb{R}^n.\)
	\end{defn}
Define
 	 	\ba
	\label{Ht}
	{\rm H}^{{l}}_k&:=& (k-{2l}+1){x}^{{2l}+1}{\rho}^{2(k-{2l})}\ddx- \frac{({2l}+1)}{2} {x}^{2l}{\rho}^{2(k-{2l})+1}\ddr  ,
	\\	\label{Rt}
	{\varTheta}^{l}_k&:=&{x}^{{2l}}\rho^{2(k-{2l})}\ddt,
	\ea
	 where
	 \(0\leqslant 2l\leqslant k.\) These vector fields are invariant  under the linear map 
	 \ba\label{zzl}
	 \mathbb{R}^3 \rightarrow   \mathbb{R}^3 : \qquad (x,y,z) \mapsto (-x,-y,-z).
	 \ea
However,	the foregoing  vector fields are \(\z\)-equivariant version of those   that are   given by  
	equations
	 \eqref{cylindFT} and \eqref{cylindTT},
	  in order to avoid  any 
 confusion, we would like to change  the notations. 
Thereby, instead of \(F\)-terms and  \(\Theta\)-terms we shall write \({\rm H}\)-terms and \(\varTheta\)-terms.

	 Denote
	\be\label{Lie}
		{\mathscr{F}}^{\mathbb{Z}_2}
	:=\Big\langle \sum a^l_{k} {\rm H}^l_k \;|\; a^l_{k}\in \mathbb{R}, {0\leqslant 2l\leqslant k,  1\leqslant k} \Big\rangle  \ee
	and
	\be\label{Lie}
	{\mathscr{T}}^{\z} :=\Big \langle \sum b^l_{k} {\varTheta}^l_k \;|\; b^l_{k} \in \mathbb{R}, {0\leqslant2l\leqslant k, 0\leqslant k} \Big\rangle. \ee
	Then we define the maximal Lie  algebra of  solenoidal conservative  \(\z\)-equivariant classical normal form of Hopf-zero vector fields  by 
\(
	{\mathscr{L}}^{\z}:= 	{\mathscr{F}}^{\z}  \oplus 	{\mathscr{T}}^{\z}.
	\)
Since
	\(	{\mathscr{L}}^{\z}\) 
is  a Lie subalgebra of \(\mathscr{L},\)  then
	\(	{\mathscr{L}}^{\z}\)   inherits
 the geometrical properties such as      conservation, incompressibility, and rotationality, from \(\mathscr{L}.\)   
	See also \cite[Theorem 2.4]{GazorMokhtariInt}.

Our next task is to present  the structure constants for \(	{\mathscr{L}}^{\z}.\)
\begin{lem} The following relations always hold. 
	\bas\label{aa} {[{\rm H}^{{l}}_k,{\rm H}^{{m}}_n]}&=&
	\left((2m+1)(k+2)-(2l+1)(n+2)\right) {\rm H}^{l+m}_{k+n},
	\\
	{[{\rm H}^{{l}}_k,{\varTheta}^{{m}}_n]}&=&\left(2m(k+2)-n(2l+1)\right){\varTheta}^{l+m}_{ k+n},
	\\
	{[{\varTheta}^{{l}}_k,{\varTheta}^{{m}}_n]}&=&0.
	\eas
\end{lem}
\begin{proof}
	The proof follows   from 
	\cite[Lemma 2.5]{GazorMokhtariInt}.
\end{proof}
The following specific cases of preceding result   will be useful 
in  the rest of the  paper.
\ba\label{STS1}
{[{\rm H}^0_1,{\rm H}^m_n]}&=&\left( 6m-n+1 \right) {\rm H}^m_{{n+1}},
\\\label{STS2}
{[{\rm H}^0_1,{\varTheta}^m_n]}&=&\left( 6m-n \right) {\varTheta}^m_{{n+1}}.
\ea

The remainder of this section is devoted to 
revisiting  a concise but detailed outline  of
the theory and method concerning   finding the  unique normal form for \({	\Gamma}\)-equivariant singularities.

\begin{itemize}
	\item  Denote \({\rm  ad}(u) {v}\)  for \({\rm  ad}(u) {v}=uv-vu,\) where \(u\) and \(v\)  are  two arbitrary vector fields. 
	\item Denote \(	{\mathscr{L}}^{\Gamma}\) for the space of all \({	\Gamma}\)-equivariant vector fields.
	\item  Denote \({{{\mathscr{L}}_k}}^{	\Gamma}\) for the space of \({	\Gamma}\)-equivariant  graded Lie algebra   with grade \(k.\)
	\({	\Gamma}\)-equivariant
	graded Lie algebra means  that   	
	\({\rm  ad}(v_i)v_j\in {{{\mathscr{L}}_{i+j}}}^{	\Gamma},\) for arbitrary \(v_i \in{{{\mathscr{L}}_i}}^{	\Gamma}\)
	and  \(v_j \in {{{\mathscr{L}}_j}}^{\Gamma}.\)
\end{itemize}
Consider the differential equation 
\(
v_s= \sum^\infty_{i=0}v_{i},
\)
where \(v_s\in 	{\mathscr{L}}^{\Gamma}.\) 
First, we define a linear map to find the first level normal form  as follows 
\bas
&	 d^{k,1}_s: 	{{{\mathscr{L}}_k}}^{	\Gamma}\rightarrow 	{{{\mathscr{L}}_k}}^{	\Gamma},&
\\
&d^{k,1}_s(Y_k):={\rm  ad}(Y_n) { v_0}.&
\eas
The  first level normal form is given by \(v^{(1)}= \sum^\infty_{i=0}v^{(1)}_{i},\) where \(v^{(1)}_i\in   \mathcal{C}^{k,1} \) 
for all \(i\)
and  \( \mathcal{C}^{k,1}\) is the complement space to \(\IM (d^{k,1}_s).\)
Note that in  order to preserve the symmetric structure  of \(v_s\), the transformations should be taken  from \(	{{{\mathscr{L}}_k}}^{	\Gamma}.\)
Proceeding inductively, we define
\bas
&d^{k, n}_s: {{\mathscr{L}}_k}^{	\Gamma}\times \ker d^{k,n-1}_s\rightarrow {{\mathscr{L}}_k}^{	\Gamma},&
\\
&d^{k,n}_s(Y^{k}_s, Y^{k-1}_s, \ldots , Y^{k-n+1}_s):= \sum^{n-1}_{i=0} {\rm  ad}( Y^{k-i}_s)v_i, \quad \hbox{ for any } n\leqslant k.&
\eas 
Then, there exists the complement subspace \(\mathcal{C}^{k,n}\) such that
\(
\IM (d^{k,n}_s)\oplus \mathcal{C}^{k,n}= {{\mathscr{L}}_k}^{	\Gamma},
\)
where \(\mathcal{C}^{k,n}\) follows the normal form style.
Then   the  \(n\)-level  normal form of \(v_s\)
is given by 
\(
w_s= \sum^\infty_{i=0}w_i,
\)
where 
\(w_i\in \mathcal{C}^{k,n}\) for all \(i.\)
For the
more detailed treatment of this theory, we refer to 
\cite{GazorYuSpec, SandersSpect,MurdBook} and 
\cite[sections 9-13]{
	sanders2007averaging}. 

\section{Normal form}\label{NFC}
In this section, we examine the unique normal form of  dynamical system \eqref{CNF}. We closely  follow  the approach in 
\cite{GazorMokhtari,BaidSand91,GazorMokhtariInt}.
 In terms of \(\rm H\)-terms and \(\varTheta\)-terms  as introduced in \eqref{Ht} 
 and \eqref{Rt}, this system  leads to the following expression
\be\label{CNFRH}
w^{(1)}:= \varTheta^0_0+b^0_1 {\varTheta}^0_1+ a^0_1 {\rm H}^0_1 +\sum a^l_k {\rm H}^l_k+\sum b^l_k {\varTheta}^l_k,\quad \hbox{for}\,\,
0\leqslant 2l\leqslant k, 1\leqslant k,
\ee
where
\(a^0_1\neq0,\)
\(a^l_k,\) and \(b^l_k\) are real constants which  could be computed  explicitly in terms of the coefficients of the original system using our  Maple program. We shall depart from  the foregoing system.
\begin{lem}
	Given the dynamical system defined by  \eqref{CNFRH},
there exists a sequence of \(\z\)-equivariant  transformations  that send  the system   to the following second level normal form
\ba\label{SLN}
w^{(2)}&:=&{\varTheta}^0_0+b^{(2)}_0 {\varTheta}^0_1+a^{(2)}_1 {\rm H}^{0}_1+\sum_{i=1}^{\infty} a^{(2)}_i{\rm H}^i_{2i}+\sum_{i=1}^{\infty} b^{(2)}_i{\varTheta}^i_{2i},
\ea
in which all of the  coefficients   are real constants.
\end{lem}
\begin{proof}
	Let \(\delta\) be grading function defined by \(\delta({\rm H}^m_n):=n,\delta({\varTheta}^m_n):=n+1.\)
Using the structure  constants given by  \eqref{STS1} and \eqref{STS2} one can deduce that
 \(\mathcal{C}^{k,2}=\Span\{{\varTheta}^0_1,{\rm H}^{m}_{2m},{\varTheta}^{m}_{2m} \,|\,  m \geqslant 1\}\) for any
  \( k \geqslant 1\)
 and the lemma follows.
\end{proof}
Similar to \cite{GazorMokhtariInt},  one can remove 
\({\varTheta}^0_0\)  using  the linear  change of variables,  for further details, see \cite[Lemma 5.3.6]{MurdBook} and \cite{GazorMokhtari}.
In what follows, we intend  to simplify the second level normal form; \(w^{(2)}.\)
First, nonetheless, we need
some notational conventions.

Suppose that there exists a non-zero  \({\rm H}^{l}_{2l}\) in \(w^{(2)}\) for some \(l\) and
there exists a non-zero  \({\varTheta}^{k}_{2k}\) in \(w^{(2)}\) for some \(k.\)
Then,
define
\bas
\mathtt{r}:= \min \{l \,|\,a^{(2)}_l\neq0, l\geqslant 1\},
\qquad\quad
\mathtt{s}:= \min \{k \,|\,b^{(2)}_k\neq0,k\geqslant 1\}.
\eas
Define the following grading function 
\bas
\delta({\rm H}^m_n):=\mathtt{r}(n-2m)+m, \,\,\,\,\, \delta({\varTheta}^m_n):=\mathtt{r}(n-2m)+\mathtt{r}+m+1.
\eas
Let us denote   the leading 
order term of \eqref{SLN} by \(\mathbb{H}_{\mathtt{r}}.\) According to the above grading function, we have that 
\ba\label{Hr}
\mathbb{H}_{\mathtt{r}}:={\rm H}^0_1+a^{(2)}_\mathtt{r}{\rm H}^{\mathtt{r}}_{2\mathtt{r}}.
\ea
This  vector field plays a prominent role in the sequel.
By performing   re-scaling \(x\) 
  in \eqref{SLN} as
\bas
x\rightarrow \left|{\frac{a_{\mathtt{r}}}{a^{(2)}_{\mathtt{r}}}}\right|^{\frac{1}{2{\mathtt{r}}}}x,\eas
then  \(\mathbb{H}_{{\mathtt{r}}}\) could be replaced 
by
\ba\label{Hr1}
\mathbb{H}_{\mathtt{r}}:={\rm H}^0_1+a_{\mathtt{r}}{\rm H}^{\mathtt{r}}_{2\mathtt{r}},
\ea
which implies that 
 \(a_r\) could be taken  an arbitrary real constant. Henceforth, for  simplicity,  we set \(a_{\mathtt{r}}=1.\)

For further reduction of \eqref{SLN}, the following lemma is quite useful.
\begin{lem}\label{trans}
	For each \({\rm H}^m_{n}\in 	{\mathscr{F}}^{\z} \)  and \({\varTheta}^m_{n}\in 	{\mathscr{T}}^{\z}\)  there exist transformations \(\mathfrak{H}^m_n\) and \(\mathfrak{R}^m_n\), such that the following hold.
		\ba\label{TH}
	[\mathfrak{H}^m_n,\mathbb{H}_{{\mathtt{r}}}]+{\rm H}^m_{n}&=&
	\frac{ ( -1 ) ^{n} ( 4m{\mathtt{r}}-2n{\mathtt{r}}+4m-n+1 )^{n-2m}_{4
 		{\mathtt{r}}+1}}{(6m-n+2)^{n-2m}_{4
 		{\mathtt{r}}+1}}
{\rm H}^{n{\mathtt{r}}+m-2m{\mathtt{r}}}_{2(n{\mathtt{r}}+m-2m{\mathtt{r}})},
	\\\label{TT}
		{[\mathfrak{R}^m_n,\mathbb{H}_{{\mathtt{r}}}]+{\varTheta}^m_{n}}&=& 
	\frac{( -1 ) ^{n}( 4m{\mathtt{r}}-2n{\mathtt{r}}+4m-n+2{\mathtt{r}}+1 )^{n-2m-1}_{4
		{\mathtt{r}}+1}}{2(6m-n+1)^{n-2m-1}_{4
		{\mathtt{r}}+1}}
	{\varTheta}^{n{\mathtt{r}}+m-2m{\mathtt{r}}}_{2(n{\mathtt{r}}+m-2m{\mathtt{r}})}.
			\ea
						\end{lem}
\bpr
Define
\bas
\mathfrak{H}^m_n&:=&\sum _{j=0}^{n-2m-1}
	\frac{( -1 ) ^{j+1} ( 4m{\mathtt{r}}-2n{\mathtt{r}}+4m-n+1 )^{j}_{4
		{\mathtt{r}}+1}}{(6m-n+2)^{j+1}_{4
		{\mathtt{r}}+1}}  {\rm H}^{j
		{\mathtt{r}}+m}_{2j{\mathtt{r}}-j+n-1},
	\\
	\mathfrak{R}^m_n&:=& \sum _{j=0}^{n-2m-1}
	\frac{( -1 ) ^{j+1}( 4m{\mathtt{r}}-2n{\mathtt{r}}+4m-n+2{\mathtt{r}}+1 )^{j}_{4
		{\mathtt{r}}+1}}{2(6m-n+1)^{j+1}_{4
		{\mathtt{r}}+1}} {\varTheta}^{j
		{\mathtt{r}}+m}_{2j{\mathtt{r}}-j+n-1}.
	\eas
	Then, substitute  the previous transformations into the left hand side of \eqref{TH} and \eqref{TT}, respectively.
	The right hand side could be obtained readily.
\epr
The following theorem is the main result of this section.
\begin{thm} \label{SIMNFT}
The  unique  normal form  of system \eqref{CNF} under assumption \(a^0_1\neq0\)  is given by 
\ba\label{SIMNF}
w^{(\infty)}&:=&
{\varTheta}^0_0+ {\rm H}^{0}_1+{\rm H}^{{\mathtt{r}}}_{2{\mathtt{r}}}+b_0 {\varTheta}^0_1+\sum_{i={\mathtt{r}}+1}^{\infty} a_i{\rm H}^i_{2i}+\sum_{i={\mathtt{s}}}^{\infty} b_i{\varTheta}^i_{2i},
\ea
or equivalently  in  the cylinder  coordinates  \(w^{(\infty)}\) takes the form 
\begin{equation*}
\left\{
\begin{aligned}
\frac{{\rm d}{x}}{{\rm d }t}&=2x\rho^2+{x}^{{2\mathtt{r}}+1}+ \sum_{i={\mathtt{r}}+1}^{\infty} a_i {x}^{{2i}+1},
\\
\nonumber
\frac{{\rm d}{\rho}}{{\rm d }t}&=-\frac{1}{2}\rho^3- \frac{({2\mathtt{r}}+1)}{2} {x}^{2\mathtt{r}}{\rho}-\sum_{i={\mathtt{r}}+1}^{\infty} \frac{({2i}+1)}{2} a_i {x}^{2i}{\rho},
\\\nonumber
\frac{{\rm d}{\theta}}{{\rm d }t}&=1+b_0\rho^2+b_\mathtt{s}{x}^{{2\mathtt{s}}}+\sum_{i={\mathtt{s}}+1}^{\infty} b_i{x}^{{2i}},
\end{aligned}
\right.
\end{equation*}
where \(b_i=0\) for all \(i\equiv_{4{\mathtt{r}}+1} (\mathtt{s}+\mathtt{r})\) and \(\mathtt{s}\not\equiv_{4{\mathtt{r}}+1}0.\)
The corresponding first integral of the unique normal form  is given by 
\ba\label{intsnf}
s^{(\infty)}(x,\rho)&:=& x\rho^4+ {\rho}^{2}  x^{2\mathtt{r}+1}+\sum_{i=\mathtt{r}+1}^{\infty} a_i  {\rho}^{2}x^{2i+1}.
\ea
	\end{thm}
\bpr
Due to special  structure constants given by equations \eqref{STS1} and \eqref{STS2} we  conclude that
   the generators of 
\(\ker(\ad_{{\rm H}^0_1})\) are \(\{{\rm H}^k_{6k+1},{\varTheta}^k_{6k}\}\) for all \(k\in \mathbb{N}.\)
Applying the method given at the end of Section \ref{HZFO} and using Lemma \ref{trans}, we have that 
\bas
{[{\rm H}^{k}_{6k+1},{\rm H}^{{\mathtt{r}}}_{2{\mathtt{r}}
}]+[\mathfrak{H}^{{\mathtt{r}}+k}_{2{\mathtt{r}}+6k+1},\mathbb{H}_{r}]}&=&
\frac{( 4{\mathtt{r}}+1) (-2k)^{4k+2}_{1}}{(4k+1)!}
{\rm H}^{4{\mathtt{r}} k+k+2{\mathtt{r}}}_{8{\mathtt{r}} k+2k+4{\mathtt{r}}}=0,
\\
{[{\rm H}^{k}_{6k+1},{\varTheta}^{{\mathtt{s}}}_{2{\mathtt{s}}}]+[\mathfrak{R}^{{\mathtt{s}}+k}_{6k+2{\mathtt{s}}+1},\mathbb{H}_{r}]}&=&
{\frac {-2{\mathtt{s}} \left( 2k+1 \right)(2{\mathtt{s}}-2k
		-8{\mathtt{r}}k)_{4{\mathtt{r}}+1}^{4k}}
	{(4{\mathtt{s}})^{4k}_{4{\mathtt{r}}+1}}}{\varTheta}^{4{\mathtt{r}}k+k+{\mathtt{s}}+{\mathtt{r}}}_{
	8{\mathtt{r}}k+2k+2{\mathtt{s}}+2{\mathtt{r}}}.
\eas
Thus, the above relations  imply  that  \({\rm H}^{k}_{6k+1}\) can not eliminate  any \({\rm H}\)-terms and 
\({\varTheta}^m_{2m}\in\IM (d^{m,{\mathtt{s}}+1}_s)\) for any  \(m\equiv_{4{\mathtt{r}}+1} ({\mathtt{r}}+{\mathtt{s}}),\)
where
\(\mathtt{s}\not\equiv_{4{\mathtt{r}}+1}0.\) 
Furthermore, one has 
\bas
{[\varTheta^{k}_{6k},{\rm H}^{{\mathtt{r}}}_{2{\mathtt{r}}}]}+{[{\mathfrak{R}}^{k+{\mathtt{s}}}_{6k+2{\mathtt{s}}
	},\mathbb{H}_{r}]} &=&
{\frac {-k \left( 4{\mathtt{r}}+1 \right) (1-2k)^{4k-1}_{1}}{({4k-1})!}
}{\varTheta}^{4{\mathtt{r}} k+k+2{\mathtt{r}}}_{8{\mathtt{r}} k+2k+4{\mathtt{r}}}=0,
\eas
which 
turns out that 
 \(\varTheta^{k}_{6k}\) generates a symmetry for the  unique  normal form of \eqref{CNF}.
This  completes the  proof.
\epr
We close this section by giving four  representations of \eqref{SIMNF} based on the  given discussions  in Section \ref{repres}.
\begin{thm}\label{4rep}
	The unique normal form  of  \eqref{CNF}  can be expressed in  the following representations.
\begin{itemize}
  	  	\item[1.] 
The unique normal form can be presented   in  the following  form 
\begin{equation*}
\left\{
\begin{aligned}
 \frac{{\rm d}{x}}{{\rm d }t}&= \frac{\partial H(x,r)}{\partial r},
\\
 \frac{{\rm d}{r}}{{\rm d }t}&=-\frac{\partial H(x,r)}{\partial x},
\\
  \frac{{\rm d}{\theta}}{{\rm d }t}&=1+b_0\rho^2+b_\mathtt{s}{x}^{{2\mathtt{s}}}+\sum_{i={\mathtt{s}}+1}^{\infty} b_i{x}^{{2i}},
\end{aligned}
\right.
\end{equation*}
    in which \(
  	H(x,r):= x r^2+ a_{\mathtt{r}} r x^{2\mathtt{r}+1}+\sum_{i=\mathtt{r}+1}^{\infty} a_i x^{2i+1} r
  	\) and \(r=\rho^2.\)
  		\item[2.]
  	Euler's form for \(y\neq0\)
  	\bas
  	w^{(\infty)}&=&
  	 \frac{1}{2} \nabla\left(x (y^2+z^2)^2\right)  \times \nabla
  	\left(\arctan ( {\frac {z}{y}} )\right)	-\frac{1}{2 }\nabla\left({ {  {y}^{2}+{z}^{2}}}\right) 
  	\times\nabla  x+\frac{b_0}{4} \nabla\left({ ( {y}^{2}+{z}^{2})^2}\right)
  	\times \nabla x
  	\\&&
  	+
  	\sum_{i={\mathtt{r}}}^{\infty} \frac{a_i}{2}
  	\nabla\left( (y^2+z^2) x^{2i+1}\right)\times\nabla
  	\left(\arctan ( {\frac {z}{y}} )\right)
  	-	\sum_{i={\mathtt{s}}}^{\infty}\frac{b_i}{2 (2i+1) }\nabla\left({ { ( {y}^{2}+{z}^{2})}}\right) 
  	\times\nabla \left(x^{2i+1}\right).
  	\eas
  	\item[3.]
  	Vector potential
  	\[
  	w^{(\infty)}={\nabla}\times {\bf A},
  	\]
  	where 
  	\bas
  	{\bf A}&=&- {\frac{1}{2}x({y}^{2}+{z}^{2} )} 
  	( z\cdot\mathbf{ {e}_{y}}-y\cdot\mathbf{ {e}_{z}})
  -
  	{ {  \frac{ 1}{2}({y}^{2}+{z}^{2} ) }}\cdot\mathbf{ {e}_{x}}	- \frac { b_0 ({y}^{2}+{z}^{2} )^2 }{4} \cdot\mathbf{ {e}_{x}}
  		\\&&
  	-\sum_{i={\mathtt{r}}}^{\infty} {\frac{a_i}{2} z}
  	{x^{2i+1}\left(\mathbf{ {e}_{y}}+{ {y}}\cdot\mathbf{ {e}_{z}}\right)}
    	-\sum_{i={\mathtt{s}}}^{\infty}
  	{ {  \frac{ b_i}{2}({y}^{2}+{z}^{2} ) {x}^{2i}}}\cdot\mathbf{ {e}_{x}}.
  	\eas
  	\item[4.] 
   Monge representation  or	Clebsch representation  for \(y\neq 0\)
  	\bas
  	w^{(\infty)}&=&-\left(\nabla \left( x^2-\frac{1}{4} (y^2+z^2)\right)
  	+ x\nabla x	\right)-\left(\nabla\left(zy ( {y}^{2}+{z}^{2})\right)
  	-{\frac {z}{y}}\nabla\left({y}^{2}\left( {y}^{2}+{z}^{2} \right) \right) \right)
  	\\&&
  	  	    -\left(\nabla\left(zy\right)-{\frac {z}{y}}\nabla\left({y}^{2}\right) \right)	+   \sum_{i={\mathtt{r}}}^{\infty} a_i\left( \frac{1}{2i}\nabla \left( { {x}^{2i+2}}\right)- x^{2i}\nabla \left(\frac{1}{4} \left( 2\,i+1 \right)  \left( {y}^{2}+{z}^{2} \right) +\frac{1}{2i}{
  		{{x}^{2}}}
  	\right)\right)
  	\\&&
  	-\sum_{i={\mathtt{s}}}^{\infty} b_i\left(\nabla\left(zy{x}^{2i}\right)-{\frac {z}{y}}\nabla\left({y}^{2}{x}^{2i}\right) \right).
  	\eas
  \end{itemize}
For all representations are  given above  the constants 	\(b_i\)  for all \(i \in \mathbb{N}_0,\) 
satisfy  the conditions that are  given in the 
{\em  Theorem  \ref{SIMNFT}}.
	\end{thm}
\bpr
Follow Corollary   \ref{Hamiltonian},  theorems \ref{Eulerform}, \ref{vectpot}, and  \ref{Clebschtheorem}, respectively.
See \cite{GazorMokhtariSandersBtriple} for the  relevant result.
\epr
\section{Practical formulas}\label{formula}
In this section, we would like to  give some fruitful formulas 
which are fundamentally 
 significant for applications.
First, some necessary relations between the coefficients  of given  Hopf-pitchfork  differential system are provided.
These relations guarantee that the classical normal form of the given   system up to third order  belongs to  \(	{\mathscr{L}}^{\z}.\)
Then, several formulas for the unique normal form's coefficients  of the  system \(w\in 	{\mathscr{L}}^{\z}\) are given. Note that  all  of the  results in this part  valid up to third order truncation.
It is quite possible to derive these computations for any finite order using our Maple program.
Finally, we  illustrate our results  with the modified Chua's circuit.

Consider the  Hopf-pitchfork  differential system governed by
\begin{equation}
\label{4NF}
\left\{
\begin{aligned}
\frac{{\rm d}{x}}{{\rm d }t}&=\sum a_{i,j,k} x^{i} y^{j}z^{k},
\\
\frac{{\rm d}{y}}{{\rm d }t}&= z+\sum b_{i,j,k} x^{i} y^{j}z^{k},
\\
\frac{{\rm d}{z}}{{\rm d }t}&=-y+\sum c_{i,j,k} x^{i} y^{j}z^{k},
\end{aligned}
\right.
\end{equation}
where  \( i+j+k=3\)  and the monomials   \(x^{i} y^{j}z^{k}\) are assumed to be   odd functions. 
	If the following relations between the coefficients of cubic terms of system \eqref{4NF} hold,
	then  the classical normal form of  this system up to third  order  belongs to \(	{\mathscr{L}}^{\z}.\)
			\ba
	\label{r1}
	a_{1,0,2}&=&-2c_{{0,2,1}}-6c_{{0,0,3}}-6b_{{0,3,0}}-2b_{{0,1,2}}-2a_{{1,
			2,0}},
	\\\label{r2}
	c_{2,0,1}&=&-b_{{2,1,0}}-3a_{{3,0,0}}.
	\ea
Following Theorem \ref{SIMNFT}, 
the unique normal  form  of \eqref{4NF} under the above  conditions is as
\begin{equation*}
\left\{
\begin{aligned}
	\frac{{\rm d}{x}}{{\rm d }t}&=a_02x\rho^2+a_1{x}^{3},
\\
\frac{{\rm d}{\rho}}{{\rm d }t}&=-\frac{a_0}{2}\rho^3-a_1 \frac{3}{2} {x}^{2}{\rho},
\\
\frac{{\rm d}{\theta}}{{\rm d }t}&=1+b_0\rho^2+b_1{x}^{4},
\end{aligned}
\right.
\end{equation*}
		where  the coefficients of previous system are   given explicitly by
	\bas
	a_0&:=&\frac{-1}{4}(c_{{0,2,1}}+3c_{{0,0,3}}+3b_{{0,3,0}}+b_{{0,1,2}}),
	\\
	b_0&:=&
	\frac{1}{8}(-3c_{{0,3,0}}-c_{{0,1,2}}+b_{{0,2,1}}+3b_{{0,0,3}}),
	\\
		a_1&:=&a_{{3,0,0}},
		\qquad
	b_1:=-\frac{1}{2}(c_{{2,1,0}}+b_{{2,0,1}}).
	\eas
	The above relations, Equation \eqref{r1}, and \eqref{r2}
 are  obtained  using 
our Maple program and employing  beneficial formulas regarding the coefficients  of the  classical normal form of  Hopf-zero singularities
derived in \cite{AlgabaHopfZ}.

 We finish with an example.
\begin{exm}
Consider modified  Chua's oscillator defined  by  the set of ordinary differential equations
\begin{equation}
\label{Chua}
\left\{
\begin{aligned}
\frac{{\rm d}{x}}{{\rm d}t}&=-\gamma x-\beta y,
\\
\frac{{\rm d}{y}}{{\rm d }t}&= z-y+x+{\mu_1}{z}^{2}y,
\\
\frac{{\rm d}{z}}{{\rm d }t}&=\alpha\left( -c z+y-a{z}^{3}
\right) +{\mu_2}{z}^{2}y.
\end{aligned}
\right.
\end{equation}
This system  possesses the  \(\z\)-symmetry given in  \eqref{zzl}.
 If one puts  
 \((\mu_1,\mu_2)=(0,0)\) in \eqref{Chua}, then  one gets the  Chua's oscillator.
	The linearization of this  system at origin  has eigenvalues 
	\(0, \pm{ i}\omega_0\) 
			when 
			
		\ba\label{Cond1}
	(c,\beta)=\Big(-\frac{  \gamma+1  }{\alpha}
	, -{\frac { \gamma\left( \alpha+\gamma+1 \right)  }{\gamma+1}}\Big),\quad\frac{ \left( \gamma+1 \right) ^{3}+\alpha \left( 2
	\gamma+1 \right) }{\gamma+1}>0,
	\ea
 where 
\(
\omega_0^2:=	 {-{\frac { \left( \gamma+1 \right) ^{3}+\alpha \left( 2
			\gamma+1 \right) }{\gamma+1}}},\) \(a\neq0,\)  \(\alpha\neq0,\) and \(\gamma\approx 0,\)   see \cite{algabachua,AlgabaHopfZ}. The following linear transformations  and rescaling the time 
		 		\bas
		 			 		x &	\rightarrow &
		 			 		\frac{\left( \gamma+1
		 			 			\right)}{ \left( \alpha+\gamma+1 \right)\gamma}\Big(
		 			 		 { { 	\omega_{{0}}\left( \gamma+1 \right)
		 			 			}}x+{ {
		 			 				\Big( \alpha+( \gamma+1 ) ^{2} \big)}}y-{\frac { \gamma \alpha}{
		 			 				\gamma+1}}z\Big),
		 		\\
		 			y 	&\rightarrow & 
		 			\frac{\left( \gamma+1
		 				\right)}{ \left( \alpha+\gamma+1 \right)}\left(
		 			{\frac {  \omega_{{0}}}{ \gamma}}x+
		 				y+z\right),
		 			\\
		 				z 	&\rightarrow& y+z,
		 				\\
		 				t&\rightarrow& \omega_0 t,
		 					\eas
		 						 			bring the  linear part of \eqref{Chua} to   the Jordan canonical form. 
		 						 			If in addition to the conditions already cited \eqref{Cond1}, we require						 			
		  \bas
	  \mu_1:=
	  {\frac {3\omega_0{\alpha}^{2}a \left( 2\alpha+ \left( \gamma+1 \right) ^{2} \right) }{ \left( 2\alpha- \left( \gamma+1 \right) ^{2}
	  		\right)  \left(\alpha -2 \left( \gamma+1 \right) ^{2}\right) }}
	 ,\qquad
  		\mu_2:=
  	 {\frac {3\omega_0{\alpha}^{2}a \left( \gamma+1 \right)  \left( 3\alpha-
  			\left( \gamma+1 \right) ^{2} \right) }{ \left( 2\alpha- \left( 
  			\gamma+1 \right) ^{2} \right)  \left(\alpha -2 \left( \gamma+1 \right) ^{2
  			} \right) }},
  			  	  \eas
  			  	  where \(\alpha \neq 2( \gamma+1 ) ^{2},\)   then the  classical normal form  of   modified Chua's oscillator up to third order belongs to \(	{\mathscr{L}}^{\z}\). 
  			  	  The previous  relations are derived 
  			  	  with the aid of  equations \eqref{r1} and \eqref{r2}.  
  			  	    			 Following Theorem \ref{SIMNFT},  the  unique  normal form  of \eqref{Chua} is given by 
  			  	  \bas
  			  	  C^{(\infty)}&:=&	
  			  	  \frac{{\alpha}^{3}a ( \alpha+
  			  	  	( \gamma+1 ) ^{2} )}{	\omega_{{0}}( \alpha+\gamma+
  			  	  	1 )\left( 2\alpha- ( \gamma+1 ) ^{2} \right) } \Big( {\frac { 3 ( \gamma+1 )
  			  	  	}{2\gamma  {
  			  	  }}}
  			  	  {\rm H}^0_1
  			  	  +
  			  	  {\frac { 2 {
  			  	  			\alpha}\gamma}{( 
  		  	  			\gamma+1 )( \alpha+\gamma+1 )   }
  		  	  {\rm H}^1_2\Big)}
  	  	  \\&&
    	  +\frac{3{\alpha}^{3}a \left(( 4\gamma+5 )  ( \gamma+1 ) ^{2}+\alpha
    	  ( 3\gamma+5 )
    	  \right)}{2 ( \alpha+\gamma+1 )\left( 2\alpha-
    	  ( \gamma+1 ) ^{2} \right)	\left( \alpha-2 ( \gamma+1 ) ^{2} \right)}  \Big(
      {\frac {( \gamma+1 ) ^{2} }{4
      	{\gamma}}}
  \varTheta^0_1
+
{\frac {{\alpha} }{    ( \alpha+\gamma+1 )}}
\varTheta^1_2\Big).
\eas
\end{exm}
We refer the reader to
\cite{GazorMokhtariInt} for relevant results.
\section{Acknowledgement}
I would like to thank  
 my  thesis advisor, Professor Majid Gazor, for the encouragement  to  write  this  paper and 
 Professor Jan A. Sanders for  invaluable remarks and comments.
\bibliographystyle{plain}

\end{document}